\documentclass[10pt,journal,compsoc]{IEEEtran}
\usepackage{amsmath,amsthm,amssymb,paralist,subfigure,cite,graphicx,color}
\usepackage{algorithm2e}
\newtheorem{lem}{Lemma}
\newtheorem{thm}{Theorem}
\newtheorem{defn}{Definition}

\newtheorem{rem}{Remark}
\newtheorem{prob}{Problem Formulation}

\def\mb{\mathbf}

\def\mc{\mathcal}

\begin{document}
	
	\title{On the Complexity of Minimum-Cost Networked Estimation of Self-Damped Dynamical Systems}
	
	\author{Mohammadreza~Doostmohammadian,~\IEEEmembership{Member,~IEEE,}
		and~Usman~Khan,~\IEEEmembership{Senior~Member,~IEEE}% <-this % stops a space
		\IEEEcompsocitemizethanks{\IEEEcompsocthanksitem M.~Doostmohammadian is with Department of Mechanical Engineering, Semnan University, Semnan, Iran.\protect\\
			E-mail: doost@semnan.ac.ir
			\IEEEcompsocthanksitem U. A. Khan is with Electrical and Computer Engineering Department, Tufts University, Medford, MA. His work has been partially supported by NSF under grants CCF-1350264, CMMI-1903972, and CBET-1935555.\protect\\
			E-mail: khan@ece.tufts.edu
		}% <-this % stops an unwanted space
		\thanks{Manuscript received Aug. 4, 2019.}}

	\markboth{IEEE Transaction on Network Science and Engineering}%
{Doostmohammadian \MakeLowercase{\textit{et al.}}: On the Complexity of Minimum-Cost Networked Estimation of Self-Damped Dynamical Systems}
	
	\IEEEtitleabstractindextext{%
		\begin{abstract}
In this paper, we consider the optimal design of networked estimators to minimize the communication/measurement cost under the networked observability constraint. This problem is known as the minimum-cost networked estimation problem, which is generally claimed to be  NP-hard. The main contribution of this work is to provide a polynomial-order solution for this problem under the constraint that the underlying dynamical system is self-damped. Using structural analysis, we subdivide the main problem into two NP-hard subproblems known as (i) optimal sensor selection, and (ii) minimum-cost communication network. For self-damped dynamical systems, we provide a polynomial-order solution for subproblem (i). Further, we show that the subproblem (ii) is of polynomial-order complexity if the links in the communication network are bidirectional. We provide an illustrative example to explain the methodologies.
		\end{abstract}

		}

		\maketitle
		\IEEEdisplaynontitleabstractindextext
		\IEEEpeerreviewmaketitle

\IEEEraisesectionheading{\section{Introduction}\label{sec_intro}}
		
\IEEEPARstart{N}{etworked} estimation has been a topic of significant interest in the literature  \cite{jstsp,Sayed-LMS,acc13,jstsp14,nuno-suff.ness,park2017design}, where a group of agents\footnote{In this paper, agent/sensor/estimator is used interchangeably.} is assigned to take measurements and share information over a communication network in order to estimate the state-vector of a dynamical system. This paper studies the complexity of Minimum-Cost Networked Estimation (MCNE) for \textit{self-damped dynamical systems}, crucial in many large-scale applications. Because of the large size, only solutions with \textit{polynomial-order complexity} are desirable. Self-damped dynamical system is a type of system in which the state of each node in the system is influenced, among others, by itself  \cite{acc13_mesbahi,trefois2015zero,chapman2015semi}.  Structurally representing this system by a digraph, each state node contains a self-cycle. Such systems are prevalent, for example, in epidemic models  \cite{nowzari2016analysis} and eco-systems  \cite{may2001book}. Also, in discretization of continuous-time systems the discretized system matrix always has non-zero diagonal entries, implying that its associated system digraph contains a self-link on every state node. Such discretized models may be derived via Euler's method or Tustin's method as discussed in \cite{kailath1980linear}. See more examples of discrete-time  representations for modeling the target-tracking systems in \cite{bar2004estimation,li2003survey,wang2011target,ennasr2018distributed,el1992generalized}. The self-damped assumption is also considered in estimation scenarios as in \cite{pequito2017structurally}.

For self-damped systems, we relax the MCNE problem into two subproblems: (i) Minimum-Cost Sensor Selection (MCSS) problem, and (ii) Minimum-Cost Communication Network (MCCN) problem. We separately discuss the computational complexity of each problem. The MCSS problem is to find the optimal sensor placement to minimize the cost of measurements. This cost may represent sensor expenses, or utility/energy consumption by sensors  \cite{rowaihy2007survey}. On the other hand, the MCCN problem is to optimally design the communication network to minimize the communication cost at the agents, where the cost may represent communication reliability  \cite{jan1993topological}, communication energy/power  \cite{wieselthier2000construction}, or distance (also referred to as capacity-infrastructure cost)  \cite{baccelli1999poisson}, among others.

\textit{Related literature:} Optimal sensor selection  \cite{pequito_gsip,jiang2003optimal} and dual problem of optimal actuator placement  \cite{tzoumas2015minimal,olshevsky2014minimal} is shown to be NP-hard\footnote{NP-hardness (Non-deterministic Polynomial-time hardness) is the defining property of a class of problems that has no solution in the time-complexity upperbounded by a polynomial function of the input parameters.} in the literature. The problem of optimal selection of sensors (information gatherers) is shown to be reducible to a \textit{minimum set covering} problem \cite{pequito_gsip}. The problem of optimal input selection is shown to be reducible to \textit{r-hitting set problem} in  \cite{tzoumas2015minimal}. These references imply that the MCSS problem is NP-hard, in general. On the other hand, cost-optimal communication network design is considered in  \cite{krause2006near,jan1993topological,wieselthier2000construction,baccelli1999poisson,pequito2017structurally,pequito2014optimal,pequito2014design}. In  \cite{krause2006near}, tradeoffs between optimal sensor placement and minimization of communication cost is claimed to be NP-hard and therefore a near-optimal solution is proposed. The near-optimal approximation\footnote{For NP-hard problems, typically a $\rho$-approximation algorithm is provided with provable guarantees on the factor $\rho$ of the returned solution to the optimal one.} solution in \cite{krause2006near} is of complexity $\mc{O}(n\log(n))$. In  \cite{baccelli1999poisson}, communication to a central unit based on Poisson-Voronoi spanning tree with application to tracking in mobile communication systems is discussed. In the literature, a few references consider the optimal communication network design under observability constraints  \cite{pequito2017structurally,pequito2014optimal,pequito2014design}; in these works, the main objective is to design the network such that the communication cost to a central base is minimized while satisfying observability constraint as a necessary condition for centralized estimation\footnote{Note that in works \cite{pequito2017structurally,pequito2014optimal,pequito2014design} although the authors claim a distributed framework, they indeed consider the estimation via a central node (or the \textit{root node}) in the sensor network.}. For example, the complexity of the cost-optimal design of the communication network for centralized estimation in  \cite{pequito2017structurally} is proved to be in $\mc{O}(n^5)$. However, the general MCCN problem is known to be NP-hard  \cite{chen1989strongly}.
%In this paper, we extend the results of our previous work  \cite{doostmohammadian2018structural} to self-damped dynamical systems and provide solution complexity of order $\mc{O}(n^2)$.

\textit{Contributions:} The main contributions of this paper are as follows, First, in Section~\ref{sec_struc}, we reformulate the MCNE problem for self-damped systems using structured systems theory and decompose the problem into MCSS and MCCN subproblems. In this direction, we generalize the optimal centralized estimation problem in  \cite{pequito2017structurally,pequito2014optimal,pequito2014design} to the networked case, where the problem is constrained with \textit{networked observability}. Second, in Section~\ref{sec_mcss},  we prove existence of a polynomial-order solution for MCSS problem in the case the system is self-damped. We reformulate this problem as a linear assignment problem with a solution of complexity $\mc{O}(N^3)$ based on the Hungarian method. It should be noted that, as claimed in  \cite{pequito_gsip}, for general systems the  MCSS problem is NP-hard. Third, in Section~\ref{sec_mccn}, we show that generally NP-hard MCCN problem has polynomial-order solutions under bidirectional link constraint, i.e., when the communication adjacency matrix is symmetric. Note that the main contribution of this work is not to generate an algorithm as the solution to the MCNE problem but determining the complexity of the solution.  %Note that the mentioned problems are claimed to be NP-hard in general, while in this work we reformulate the problems into graph-theoretic problems with P-order algorithms as their solutions.
{It should be emphasized that in this paper, we determine the minimum number of agents such that each agent measures one necessary state for observability of the underlying self-damped dynamical system. In other words, the case of \textit{minimal system observability} is considered here. The case where more measurements of the system are given, for example, to improve the estimation performance or to reduce the cost of network communication, is left as future research direction.}

\section{MCNE Problem Statement} \label{sec_prob}
We consider discrete-time LTI dynamics in the form:
\begin{eqnarray}\label{eq_sys1}
\mb{x}(k+1) = A\mb{x}(k) + \mb{v}(k) \\
\mb{y}(k) = C\mb{x}(k) + \mb{r}(k)
\end{eqnarray}
where $k$ is the time-index, $\mb{x}\in\mathbb{R}^n$ represents the state of the dynamical system, $\mb{y}\in\mathbb{R}^N$ is system measurement, $\mb{v}$ and $\mb{r}$ respectively represent system and measurement noise\footnote{{The noise is inherent to any estimation scenario. Although, in this paper we do not directly use the noise in our analysis, we consider the noise terms as in most general case the concept of estimation and observability are tightly related with noise. It should be noted that the cost-optimal design in this paper is irrespective of the noise terms. However, after designing the sensor-network, the distributed estimation scenario aims to track the noise-corrupted system states.}}, and $A$ and $C$ respectively represent system and measurement matrices. {The discrete-time model \eqref{eq_sys1} may be derived from the discretization of continuous-time models in the form,
\begin{eqnarray}\label{eq_sys_cont}
\dot{\mb{x}} = \bar{A}\mb{x}
\end{eqnarray}
Applying the \textit{Euler's method} \cite{kailath1980linear} for discretization of \eqref{eq_sys_cont} with sampling time $T$,	

\begin{eqnarray}\label{eq_sys_disc_euler}
\mb{x}((k+1)T) \simeq (I+T\bar{A})\mb{x}(kT)
\end{eqnarray}
where $I$ is the identity matrix. In fact, the Euler's method is the approximation to the following discrete-time model of continuous-time system \eqref{eq_sys_cont}:
\begin{eqnarray}\label{eq_sys_disc}
\mb{x}((k+1)T) = \exp(T\bar{A})\mb{x}(kT)
\end{eqnarray}
Another approximation to the above model is by using the \textit{Tustin's method} \cite{kailath1980linear} for discretization of \eqref{eq_sys_cont} as follows:
\begin{eqnarray}\label{eq_sys_disc_tustin}
\mb{x}((k+1)T) \simeq (I-\tfrac{T}{2}\bar{A})^{-1}(I+\tfrac{T}{2}\bar{A})\mb{x}(kT)
\end{eqnarray}

In both discretized models \eqref{eq_sys_disc_euler} and \eqref{eq_sys_disc_tustin} the diagonal entries of the discrete-time system matrices are non-zero (due to the identity matrix $I$).}  It should be mentioned both Euler's and Tustin's methods are discrete-time approximations of continuous-time system model, and both give  approximate solutions to the continuous-time model \eqref{eq_sys_cont}.

In general, system estimation necessitates the pair $(A,C)$ to be observable\footnote{{Throughout this paper, we use structured system theory where the observability analysis is structural and generic. For notation simplicity, observability implies structural or generic observability.}}. In networked estimation, the group of agents/estimators are connected such that the system is observable to every agent/estimator via the local measurement matrices~$C_i$ at each agent~$i$ with~$C=[C_1^\top\ldots C_N^\top]^\top$. Since the pair $(A,C_i)$ is not necessarily observable at any agent, the agents recover the observability deficiency by sharing measurements or state predictions over a communication network $\mc{G}_U$. To keep the exposition simple and without loss of generality, we assume each measurement $\mb{y}_i$ is taken by one sensor/agent $i$. At every time-step $k$, every agent $i$ shares its information with other agents in its neighborhood $\mc{N}_i$. By sharing necessary information over $\mc{G}_U$ every agent is capable of tracking the global state of the dynamical system. In this regard, the necessary condition for networked estimation is \textit{networked observability} defined as follows~\cite{acc13,jstsp14,nuno-suff.ness,park2017design,globalsip14}:
\begin{defn}\label{def_netobs}
	For the dynamical system $(A,\{C\}_i)$ monitored by a network $\mc{G}_U$ of agents with adjacency matrix $U$, the networked system is observable by each agent if the pair $(U \otimes A, D_C)$ is observable, where $D_C$ is defined as:
	\begin{eqnarray}
	D_C \triangleq \left(
	\begin{array}{cccc}
	\sum_{j\in \mathcal{N}_1}C_j^\top C_j\\
	&\ddots\\
	& &\sum_{j\in \mathcal{N}_N}C_j^\top C_j\
	\end{array}
	\right)
	\end{eqnarray}
	and the set $\{\mc{N}_1,\hdots,\mc{N}_N\}$ is the set of neighboring agents, where $\mc{N}_j$ is the set of agent $j$'s neighbors.
\end{defn}
It should be noted that the $(A,C)$-observability is necessary for $(U \otimes A, D_C)$-observability, but is not sufficient; the next section explains more conditions for networked observability. If the networked observability is satisfied, a feedback gain matrix may exist such that every agent achieves asymptotic omniscience on the dynamical system state  \cite{acc13,jstsp14,nuno-suff.ness,park2017design,globalsip14}. This simply implies that the error dynamics at every agent is bounded and achieves asymptotic mean-squared stability. The idea in this paper is to design the measurement matrix and communication network such that certain cost is minimized. The cost of networked estimation is twofold: the measurement cost and the communication cost. The problem is to minimize these costs while networked observability constraint is satisfied, termed as the Minimum-Cost Networked Estimation (MCNE) problem.
\begin{prob} \label{prob_1}
	Consider the system matrix and measurement matrix pair $(A,C)$. The measurement of state $j$ by agent $i$ is assigned with a cost $\delta_{ij}$ and the communication from agent $i$ to agent $j$ is assigned with a cost $\eta_{ij}$. Then, the problem is to solve the following:
	\begin{equation} \label{eq_formulate_1}
	\begin{aligned}
	\displaystyle
	\min
	~~ & \sum_{i=1}^{N} \sum_{j=1}^{n} \delta_{ij}\mc{C}_{ij} +\sum_{i=1}^{N} \sum_{j=1}^{n} \eta_{ij}\mc{U}_{ij} \\
	\text{s.t.} ~~ & (U\otimes A,D_C)\mbox{-observability},\\
	~~ &  \mc{C}_{ij} \in \{0,1\}\\
	~~ &  \mc{U}_{ij} \in \{0,1\}\\
	\end{aligned}
	\end{equation}
	where the matrices $\mc{C} \sim \{0,1\}^{N \times n}$ and $\mc{U} \sim \{0,1\}^{N \times N}$ represent the $0-1$ structure of $C$ and $U$ matrices, respectively.	
\end{prob}

\begin{rem} \label{mcne}
	The general MCNE problem is NP-hard to solve.
\end{rem}
We prove this remark in the next sections as we reformulate the problem using structural analysis.
In this paper, we solve the MCNE problem for self-damped systems.

\begin{defn}\label{defn_self}
	Self-damped systems are where the evolution of every state $\mb{x}_i$ is a function of, among other states, the state $\mb{x}_i$ itself. In structured systems theory, a self-damped system is modeled by a graph with a self-loop on every state node \cite{acc13,jstsp14,nuno-suff.ness,park2017design,globalsip14} \cite{acc13_mesbahi,trefois2015zero}.
\end{defn}
{The self-damped system dynamics is prevalent in discretized representations of continuous-time systems, as in the mentioned Euler's and Tustin's discretized models.} \textit{The main problem addressed in this paper is on the complexity of MCNE problem under self-damped system constraint. We investigate if there is a polynomial-order solution for this problem; and if not, is there an efficient (with polynomial-order complexity) approximate algorithm to solve the problem, and what is the $\rho$-approximation of the solution?} To answer these questions, in the next sections,  using structured systems theory we reformulate the problem for self-damped systems into two subproblems and find the complexity of the solution for each subproblem.

{\textit{Assumptions:} The following assumptions are made throughout the paper:
\begin{enumerate}[(i)]
\item The underlying system to be estimated is self-damped.
\item The system matrix $A$ is not necessarily \textit{irreducible}\footnote{A reducible matrix $A$ is such that, by simultaneous row/column permutations, it can be transformed into block upper/lower-triangular form. Otherwise, it is irreducible \cite{rein_book}.}, see the comments after Lemma \ref{lem_nec}.
\item Minimum number of measurements for $(A,C)$-observability are available.
\item The communication links in the  network of agents/sensors are bidirectional.
\item Each agent measures one system  state. 
%For optimal sensing cost, no two agents share measurement from the same parent SCC, see Section~\ref{sec_rem}.
\end{enumerate}
Regarding Assumption (v), it should be mentioned that the methodology can be extended to the case where agents take more than one measurements, discussed in Sections~\ref{sec_mccn} and~\ref{sec_con}.

\section{Reformulation Based on Structured Systems Theory} \label{sec_struc}
In this paper, a structural (also referred to as \textit{generic}) approach is adopted to solve the MCNE problem. It is known that many properties of the system emerge from the system structure and are irrespective of the numerical values of system variables  \cite{woude:03}. Among these properties are system controllability and observability  \cite{liu-pnas,globalsip14,asilomar11}. {It is known that  the values of system parameters for which a generic property does not hold lies  on an algebraic subspace with \textit{zero Lebesgue measure} \cite{woude:03}. This implies that the structural observability results in the observability for \textit{almost all} values of system/measurement parameters. Determination of the zero-measure algebraic subspace for which structural observability does not imply observability is case-specific and generally can be formulated based on the dependencies of system parameters and is out of scope of this paper.}

In the rest of the paper, consider matrix $\mc{A} \sim \{0,1\}^{n \times n}$ as the $0-1$ structure of system matrix $A$ and $\mc{C} \sim \{0,1\}^{N \times n}$ as the $0-1$ structure of measurement matrix $C$. Based on the structured systems theory, $\mc{A}$ and $\mc{C}$ can be represented by a graph known as \textit{system digraph}. The $(A,C)$-observability generically emerges from this graphical model. For self-damped systems, the system digraph contains a self-loop on every state node. Mathematically, this implies that $\mc{A}_{ii}=1$ for all $i \in \{1, \hdots,n\}$, and therefore, it is known that self-damped systems are structurally full-rank  \cite{asilomar11}. For such systems, the observability can be analyzed via certain components in the system digraph as we discuss below.
\begin{defn}
	In a system digraph, a component in which every state has a path to every other state in the same component is called a \textit{Strongly Connected Component (SCC)}. A SCC with no outgoing links to other SCCs is called as \textit{parent SCC}, denoted by $\mc{S}^p$, and a non-parent SCC is called \textit{child SCC}. The partial order of SCCs is denoted by $\prec$, i.e., $\mc{S}_i \prec \mc{S}_j$ implies that there is a directed path from $\mc{S}_i$ to $\mc{S}_j$. {Define the set $\mc{S}^p=\{\mc{S}^p_1,\mc{S}^p_2,\dots\}$ as the set of all parent SCCs.}
\end{defn}
\begin{rem} \label{scc}
	The algorithm for decomposing a system digraph into SCCs and determining their partial order (parent-child classification) is called \textit{depth-first-search} algorithm with computational complexity $\mc{O}(n^2)$  \cite{cormen2009introduction}, with $n$ as the number of graph nodes.
\end{rem}

The following theorem relates the observability of self-damped systems with SCCs in their system digraph, in a generic sense.

\begin{thm} \label{thm_Sp}
	A self-damped system digraph is observable {if and only if} for each parent SCC, $\mc{S}^p_j$, there is one state node $x_i \in \mc{S}^p_j$ measured by an agent.
\end{thm}
\begin{proof}
The proof follows the main theorem on structural observability developed in  \cite{lin,rein_book}. Based on this theorem, two necessary and sufficient conditions on the system digraph for structural observability are as follows: (i) there is a directed path form every state node to an output (or measurement), and (ii) there is a family of disjoint cycles spanning all nodes in the system digraph\footnote{The condition (i) for structural observability is known as \textit{output connectivity condition}, and condition (ii) is known as \textit{rank condition}.}.

{\textit{Sufficiency}:} Based on the definition of the self-damped systems the condition (ii) is already satisfied. Note that based on the definition, there is a path from every child SCC to (at least) one parent SCC. Having an output $\mb{y}_i$ from one state in every parent SCC, $\mc{S}^p_j$, implies the output-connectivity of all state nodes in the same parent SCC, and further, all states in the child SCCs connected to $\mc{S}^p_j$ via a directed path, i.e., every $\mc{S}_i^c$ for which $\mc{S}^c_i \prec \mc{S}^p_j$ is also connected to the output $\mb{y}_i$. This holds for every parent and child SCC. This satisfies the condition (i) for structural observability and the theorem follows.

{\textit{Necessity}: We prove the necessity by contradiction. Consider the case where (at least) one parent SCC, say $\mc{S}^p_i$, has no outgoing measurement. Therefore, the node states in $\mc{S}^p_i$ are not connected to any output. This is because (i) no agents measure a state node in $\mc{S}^p_i$, and (ii) based on the definition of parent SCC, there is no path from states in $\mc{S}^p_i$ to any other output-connected SCC. This implies that the output-connectivity is not satisfied, and therefore the system is not observable.
	}
\end{proof}

Next, we extend the observability results to networked estimation acquired by a network of agents. In this sense, the network must be specifically designed to ensure networked observability as follows.

\begin{thm} \label{thm_net}
	Let a self-damped system have all the measurements  for structural observability at the agents (Theorem~\ref{thm_Sp}). For the networked estimation protocol to achieve asymptotic omniscience on system state (networked observability according to Definition~\ref{def_netobs}), the network $\mc{G}_U$ {is sufficient to} be Strongly Connected (SC).
\end{thm}
\begin{proof}
	The proof outline is similar to the proof of Theorem~\ref{thm_Sp}. To satisfy the networked observability (according to the definition), the pair $(U \otimes A, D_C)$ must be observable. Following the structural observability in  \cite{lin,rein_book}, for the self-damped underlying system, every agent applies self information for networked estimation (along with the information of neighbors). This implies that the matrix $U \otimes A$ is structurally full-rank and the rank condition for structural observability is satisfied  \cite{acc13,jstsp14,nuno-suff.ness,park2017design,asilomar11}. In the networked system graph associated to $U \otimes A$, the Strong Connectivity of $\mc{G}_U$ implies that access to the measurements/outputs is shared among all agents via a path. In the networked system graph including self-damped sub-systems, according to Theorem~\ref{thm_Sp} every parent SCC, say $\mc{S}^p_j$, is output-connected. Assume agent $i$ takes (at least) one state measurement $\mb{y}_i$ in $\mc{S}^p_j$. The Strong-connectivity of $\mc{G}_U$ implies that there is a directed path from every agent, say $k$, to agent $i$. Therefore, all the states in the sub-system associated to the agent $k$ are connected via this path to the output $\mb{y}_i$ measured by the agent $i$. This holds for all agents measuring a state in parent SCCs, and therefore the output-connectivity of all parent SCCs follows from strong-connectivity of $\mc{G}_U$. The output-connectivity of child SCCs follows from the similar argument as in Theorem~\ref{thm_Sp}, and the output-connectivity condition of the structural observability theorem in  \cite{lin,rein_book} follows. { This implies that SC network among agents is sufficient for networked estimation/observability.}
\end{proof}
{
\begin{lem}\label{lem_nec}
	Assuming each agent takes one measurement, the minimum number of agents to estimate the state of a self-damped dynamical system is equal to $|\mc{S}_p|$, where $|\cdot|$ is the cardinality of the set. Further, for this minimum number of agents, the SC network among agents is necessary for networked observability/estimation\footnote{{Note that here without loss of generality we assume every agent measures the states in one parent SCC. In case the number of agents is less than $|\mc{S}_p|$, some agents measure the states in more parent SCCs. Assuming no two agents take measurement from the same parent SCC we still need an SC network among these agents. In fact, the key point here is that the information of every parent SCC reaches to every other agent via a directed path.}}.
\end{lem}	
\begin{proof}
    The proof follows the proof of Theorems~\ref{thm_Sp} and~\ref{thm_net}. Note that measuring a state node in every parent SCC is necessary and sufficient for observability. Assuming every agent takes one state measurement, the minimum number of agents to satisfy observability is equal to the number of parent SCCs, i.e., $|\mc{S}^p|$. Next, assume we have the minimum number of agents each measuring a state in  a parent SCC. In the network estimation scenario, having an SC network each agent's information (regarding the parent SCC measured by that agent) reaches to every other agent via a directed path. This implies that in the networked system every parent SCC is observable to every agent. Let assume that the communication network is not SC. This implies that (at least) there is no directed path from one agent, say $a$, to another agent, say $b$. Therefore, the information of parent SCC $\mc{S}^p_i$ measured by agent $a$ cannot reach to agent $b$. Note that we have the minimum number of agents/measurements and, therefore, no other agent is measuring any state node in $\mc{S}^p_i$. This implies that the states in $\mc{S}^p_i$ are not observable to agent $b$ and the networked observability is violated. Therefore, for minimum number of agents, the networked estimation error cannot attain steady-state stability over a non-SC network.
\end{proof}
}
{
Note that the networked observability results in this section are particularly defined for non-SC system digraphs, i.e., the case system matrix~$A$ is reducible. In case the system digraph is SC, according to Theorem~\ref{thm_Sp}, only one measurement is necessary and sufficient for structural observability. Therefore, only one agent may perform the estimation and the concept of networked observability is irrelevant. This justifies Assumption (ii) in this paper. In case having more than one agent, measuring perhaps different state nodes in the SC system digraph, there is no need for the communication network of agents to be SC and it  might be even disconnected.}

{It should be mentioned that, following the same line of justification as in Lemma~\ref{lem_nec}, the necessary SC network condition can be extended to the case where agents take two or more distinct measurements. If no two agents share a measurement of the same parent SCC, the SC network among agents is necessary for networked observability.  The Strong Connectivity of multi-agent communication network is a typical assumption in  networked estimation literature as in  \cite{jstsp,Sayed-LMS,acc13,jstsp14,nuno-suff.ness,park2017design,asilomar11}, and also in optimal design of sensor networks as in \cite{pequito2017structurally}. In this paper, we consider networked estimation via the minimum number of measurements each from a parent SCC. This is to minimize the measurement costs.  Based on the results of Theorem~\ref{thm_Sp}, Theorem~\ref{thm_net} and Lemma~\ref{lem_nec}, assuming networked estimation with minimum number of measurements defined in Lemma~\ref{lem_nec} each assigned to one agent,} one can relax the networked observability constraint and reformulate the MCNE problem \eqref{eq_formulate_1} for self-damped systems into two subproblems as follows:
\begin{prob}
	Consider the setup in Problem Formulation~\ref{prob_1} for a self-damped system. The problem can be subdivided into Minimum-Cost Sensor Selection (MCSS) problem:
	\begin{equation} \label{eq_formulate_2sensor}
	\begin{aligned}
	\displaystyle
	\min
	~~ & \sum_{i=1}^{N} \sum_{j=1}^{n} \delta_{ij}\mc{C}_{ij} \\
	\text{s.t.} ~~ & (A,C)\mbox{~observability},\\
	~~ &  \mc{C}_{ij} \in \{0,1\}\\
	~~ &  \mc{A}_{ii}=1,~~ \forall i \in \{1, \hdots, n\}  \\
	\end{aligned}
	\end{equation}
	and Minimum-Cost Communication Network (MCCN) problem:
	\begin{equation} \label{eq_formulate_2net}
	\begin{aligned}
	\displaystyle
	\min
	~~ & \sum_{i=1}^{N} \sum_{j=1}^{n} \eta_{ij}\mc{U}_{ij} \\
	\text{s.t.} ~~ & \mc{G}_U\mbox{~is~SC},\\
	~~ &  \mc{U}_{ij} \in \{0,1\}\\
	\end{aligned}
	\end{equation}
\end{prob}
In order to justify the above formulation, note that based on Theorems~\ref{thm_Sp} and \ref{thm_net} and Lemma~\ref{lem_nec}, the networked observability constraint in Problem Formulation~\ref{prob_1} can be decomposed into (i) $(A,C)$-observability constraint for which every parent SCC is measured by (at least) one agent related to the optimization term $\sum_{i=1}^{N} \sum_{j=1}^{n} \delta_{ij}\mc{C}_{ij}$, and (ii) strong-connectivity of multi-agent network related to the optimization term $\sum_{i=1}^{N} \sum_{j=1}^{n} \eta_{ij}\mc{U}_{ij}$. Further we note that,
\begin{eqnarray} \nonumber
\min \left(\sum_{i=1}^{N} \sum_{j=1}^{n} \delta_{ij}\mc{C}_{ij} +\sum_{i=1}^{N} \sum_{j=1}^{n} \eta_{ij}\mc{U}_{ij} \right) \\
= \min \sum_{i=1}^{N} \sum_{j=1}^{n} \delta_{ij}\mc{C}_{ij} + \min \sum_{i=1}^{N} \sum_{j=1}^{n} \eta_{ij}\mc{U}_{ij}
\end{eqnarray}
This is because both summations (including the weights $\delta_{ij}$ and $\eta_{ij}$) are positive, therefore the minimization of the sum is equivalent to the minimization of each term. {The $(A,C)$-observability constraint in  \eqref{eq_formulate_2sensor}, which according to Theorem~\ref{thm_Sp} implies that one state node in every parent SCC must be measured,  is related to MCSS problem. On the other hand, the constraint $\mc{G}_U$ being SC in \eqref{eq_formulate_2net}, according to Theorem~\ref{thm_net} and Lemma~\ref{lem_nec}, is related to MCCN problem.} Notice that the constraint $\mc{A}_{ii}=1,~~ \forall i \in \{1, \hdots, n\}$ formulates the self-damped system constraint.  From these arguments, the MCNE problem is decomposed into MCSS and MCCN optimization problems separately discussed in the next sections. {Based on the structured systems theory, with the given assumptions, the optimality and complexity of the MCNE problem is \textit{almost always} the same as the MCSS and MCCN problems. We should emphasize that this decomposition is only valid for self-damped systems, and for general systems such decomposition might be irrelevant. }

\section{MCSS Problem: Algorithm and Complexity} \label{sec_mcss}
Recall that the MCSS problem is the problem of identifying the states to be measured  such that  a certain cost of measurements  is minimized while satisfying  observability condition for inference purposes.

\begin{rem} \label{mcss}
	For general systems the MCSS problem is NP-hard.
\end{rem}
Note that for general systems (not necessarily self-damped), the MCSS problem is proved to be reducible to \textit{minimum set covering} problem and therefore is NP-hard  \cite{pequito_gsip}. In this section, we find a polynomial-order solution under self-damped system constraint. First, we add two new constraints on the state-measurement pairs. For minimization purposes, we assume that each agent is assigned to measure one and only one state, implying that each~$C_i$ is a row vector and $\sum_{j=1}^{n} \mc{C}_{ij} = 1$. Also, each state is at most measured by one agent, implying that $\sum_{i=1}^{N} \mc{C}_{ij} \leq 1$. Adding these conditions the new MCSS formulation is as follows:
\begin{prob}
	Considering that every agent measures only one state, the MCSS problem has the following form:
	\begin{equation} \label{eq_MCSS}
	\begin{aligned}
	\displaystyle
	\min
	~~ & \sum_{i=1}^{N} \sum_{j=1}^{n} \delta_{ij}\mc{C}_{ij} \\
	\text{s.t.} ~~ & (A,C)\mbox{~observability},\\
	~~ &  \mc{C}_{ij} \in \{0,1\}\\
	~~ &  \sum_{i=1}^{N} \mc{C}_{ij} \leq 1\\
	~~ &  \sum_{j=1}^{n} \mc{C}_{ij} = 1\\
	~~ &  \mc{A}_{ii}=1,~~ \forall i \in \{1, \hdots, n\}  \\
	\end{aligned}
	\end{equation}
\end{prob}
Using the results of Theorem~\ref{thm_Sp}, the $(A,C)$-observability constraint can be relaxed as having one state measurement from each parent SCC $\mc{S}^p_j$ to be assigned with an agent~$i$. The agent measures the state in parent SCC that has minimum measurement cost.
In this direction, redefine the state measurement cost matrix $\delta$ by a new parent SCC measurement cost matrix $\Delta$ as follows:
\begin{eqnarray} \label{eq_Delta}
\Delta_{ij} = \min(\delta_{im}), ~ x_m \in \mc{S}^p_j
\end{eqnarray}
Note that the new cost matrix $\Delta$ is $N \times N$. Therefore, instead of $\mc{C}$, a new assignment matrix needs to be defined relating the parent SCCs to agents.  Denote this $0-1$ matrix by $\mc{Z}=\{\mc{Z}_{ij}\}$. An element $\mc{Z}_{ij} = 1$ implies that agent $i$ is assigned to measure the minimum-cost state in parent SCC $j$ and $\Delta_{ij}$ denotes this cost. Recalling that measuring all $N$ parent SCCs guarantees the $(A,C)$-observability (Theorem~\ref{thm_Sp}) and the fact that parent SCCs do not share state nodes  \cite{cormen2009introduction}, the new optimization formulation is as follows:
\begin{prob} \label{prob_Z}
	Redefining the agent-SCC cost matrix $\Delta$ and introducing the agent-SCC assignment matrix $\mc{Z}$, the MCSS problem has the following form:
	\begin{equation} \label{eq_LSAP}
	\begin{aligned}
	\displaystyle
	\min
	~~ & \sum_{i=1}^{N} \sum_{j=1}^{n} \Delta_{ij} \mc{Z}_{ij} \\
	\text{s.t.} ~~ &  \mc{Z}_{ij} \in \{0,1\}\\
	~~ &  \sum_{i=1}^{N} \mc{Z}_{ij} = 1\\
	~~ &  \sum_{j=1}^{n} \mc{Z}_{ij} = 1\\	
	\end{aligned}
	\end{equation}
\end{prob}
Note that $\sum_{i=1}^{N} \mc{Z}_{ij} = 1$ implies that each parent SCC is measured by one agent, and $\sum_{j=1}^{N} \mc{Z}_{ij} = 1$ implies that each agent makes one measurement of a parent SCC.
The above formulation is a \textit{linear assignment problem}, which is well-known in combinatorial optimization. This problem is discussed in the literature  to a great extent. For extensive surveys on this problem and generalizations see  \cite{cattrysse1992survey,burkard1999linear}. The most well-known polynomial-order solution for linear assignment problem is the \textit{Hungarian method}  \cite{kuhn1955hungarian}. {The pseudo-code for the Hungarian method is given in Algorithm~\ref{alg_hung}.} The computational complexity of this algorithm is $\mc{O}(N^3)$. Recalling that the formulation~\eqref{eq_LSAP} is equivalent to the formulation~\eqref{eq_formulate_2sensor} leads to the following remark,
\begin{rem} \label{mcss2}
	The computational complexity of MCSS problem solution for self-damped system is $\mc{O}(N^3)$, where $N$ is the number of agents (or parent SCCs).
\end{rem}

\begin{algorithm}[!t] \label{alg_hung}
	{
	\textbf{Given:} Cost matrix  $\Delta=[\Delta_{ij}]$ \;
	\For{$i=1,\hdots,N$}{
		$u_i =$ smallest entry in row $i$ of $\Delta$\;
		\For{$j=1,\hdots,N$}{$\hat{\Delta}_{ij}=\Delta_{ij}-u_i$\	
		}
	}
	\For{$j=1,\hdots,N$}{
		$v_j =$ smallest entry in column $j$ of $\hat{\Delta}$\;
		\For{$i=1,\hdots,N$}{$\hat{\Delta}_{ij}=\hat{\Delta}_{ij}-v_j$\	
		}
	}
	$S =$ an independent set of zeros of max size in  $\hat{\Delta}$\; 	
	$q = |S|$ \;
	\While{$q<N$}{
		Cover $\hat{\Delta}$\;
		$k = $ smallest entry in $\hat{\Delta}$ not covered by a line\;
		\For{$i=1,\hdots,N$}{
			\For{$j=1,\hdots,N$}{\If{$\hat{\Delta}_{ij}$ is not covered}{$\hat{\Delta}_{ij} = \hat{\Delta}_{ij} - k$ }
				\If{$\hat{\Delta}_{ij}$ is covered twice}{$\hat{\Delta}_{ij} = \hat{\Delta}_{ij} + k$ }	
			}
		}	
		$S =$ an independent set of zeros of max size in  $\hat{\Delta}$\;
		$q = |S|$ \;}
	\For{$i=1,\hdots,N$}{
		\For{$j=1,\hdots,N$}{
			\If{$\hat{\Delta}_{ij} \in S$ }{$\mathcal{Z}_{ij} = 1$}
			\Else{$\mathcal{Z}_{ij} = 0$}
		}
	}
	\textbf{Return} $\mathcal{Z}=[\mathcal{Z}_{ij}]$\;\
}	
	\caption{Pseudo-code for the Hungarian Algorithm}
\end{algorithm}

\section{MCCN Problem: Algorithm and Complexity} \label{sec_mccn}
Recall that MCCN problem is to find the minimum weight (cost) strongly-connected subgraph spanning all nodes (agents) in the communication network.
\begin{rem} \label{mccn}
	For general (directed) communication networks the MCCN problem is NP-hard.
\end{rem}
This is because the MCCN problem is reducible to  \textit{directed Hamiltonian cycle} problem and therefore is NP-hard   \cite{eswaran1976augmentation,chen1989strongly}. This problem is also known as \textit{minimum spanning strong sub(di)graph} in literature  \cite{bang2008minimum}. For approximation algorithms to this NP-hard problem,  \cite{khuller1995approximating,khuller1996strongly} provide a $1.62$-approximation algorithm, and  \cite{vetta2001approximating} proposes a $1.5$-approximation algorithm. We consider an undirected communication network among agents, i.e., the communication links are all bidirectional. This simply implies that if two agents are in the communication range of each other, e.g., in a wireless sensor network, both agents share their information. This is a typical assumption in the literature of networked estimation as in  \cite{nuno-suff.ness,Sayed-LMS}. This assumption changes the problem as in the following:
\begin{prob} Considering bidirectional communication among agents, the MCCN problem has the following form:
	\begin{equation} \label{eq_MCCN_bi}
	\begin{aligned}
	\displaystyle
	\min
	~~ & \sum_{i=1}^{N} \sum_{j=1}^{n} \eta_{ij}\mc{U}_{ij} \\
	\text{s.t.} ~~ & \mc{G}_U\mbox{~is~SC},\\
	~~ &  \mc{U}_{ij} \in \{0,1\}\\
	~~ &  \mc{U} \mbox{~is~symmetric}\\
	\end{aligned}
	\end{equation}
	\label{prob_mccn}
\end{prob}
The above problem can be reformulated as a well-known problem in combinatorial optimization and discrete mathematics, known as \textit{minimum weight spanning tree}. Two classic polynomial-order solutions (with complexity $\mc{O}(N^2\log(N))$) for this problem are \textit{Prim} algorithm  \cite{prim1957shortest} and \textit{Kruskal} algorithm  \cite{kruskal1956shortest}. However, a more efficient distributed algorithm with computational complexity $\mc{O}(N\log(N))$ is proposed in  \cite{gallager1983distributed}. {The pseudo-code for the distributed algorithm is given in \cite{gallager1983distributed} and excluded here due to space limitation.} Recalling that the formulation~\eqref{eq_MCCN_bi} is the equivalent form of the formulation~\eqref{eq_formulate_2net}, we deduce the following.
\begin{rem} \label{mccn2}
The computational complexity of the most efficient solution to MCCN problem for undirected networks (with bidirectional links) is $\mc{O}(N\log(N))$.
\end{rem}

{
\subsection{Remarks on SC Communication Network Condition} \label{sec_rem}
Note that with the help of Assumptions (i)-(v), we show in Lemma~\ref{lem_nec} that SC communication network among agents is necessary for networked estimation. However, in general, for networked estimation there might be cases for which some of the given assumptions are violated and therefore the agents' network is not necessarily SC. Assume that we are interested to reduce the number of communications among agents by, for example, increasing the number of system measurements taken by agents.
%This is also the case where the number of agents is more than minimum number defined by Lemma~\ref{lem_nec}. In such cases, the communication network of agents is not necessarily SC. However, the same amount of communications over the network is required, as every agent needs the information of every parent SCC via a directed path. To illustrate more,
In this direction, consider three cases:
\begin{itemize}
\item
Case (I):
Following Assumption (v) let each agent take one measurement. Consider the number of agents to be $N_1$ and the number of parent SCCs to be $N_2<N_1$. In the communication network, every agent needs to receive the information of the other $N_2 -1$ parent SCCs via directed paths. In such case, although the network is not necessarily SC, the amount of communications is more than the case where $N_2$ agents each measure one parent SCC and share information over (smaller) SC network.
\item
Case (II): consider $N_2$ measurements each from one parent SCC are assigned to $N_1<N_2$ agents, implying that some agents take more than one measurement. Since no two agents share a measurement and following the same reasoning as in Lemma~\ref{lem_nec}, the SC communication network is a necessary condition and MCCN problem formulation~\eqref{eq_MCCN_bi} follows.
\item
Case (III): consider $N_1$ measurements more than necessary $N_2<N_1$ parent SCC observations. Let us assign these measurements to $N_3<N_1$ agents, where some agents may share measurements from one or more parent SCCs. In this case, the minimum communication network is not necessarily SC and could be disconnected. Therefore, recalling the bidirectional link assumption among agents, the SC condition on $\mc{G}_U$ is relaxed to having a disconnected group of smaller SC sub-networks. Recall that the solution to the MCCN problem subject to SC undirected network condition is shown to be minimum weight spanning tree. According to the definition, removing any link from a spanning tree yields a disconnected network of smaller trees knwon as a \textit{forest} \cite{algorithm}. Therefore, one may run similar algorithm over  SC sub-networks and find the \textit{minimum weight spanning forest} as the solution. See more information in \cite{matsui1995minimum}.
\end{itemize}
To summarize, the communication cost in the Case~(I) is not less than the MCCN problem formulation~\eqref{eq_MCCN_bi} while the measurement cost is more than the MCSS problem formulation~\eqref{eq_MCSS}. Case~(II) can be considered as an extension to the  MCCN problem formulation~\eqref{eq_MCCN_bi} and the MCSS problem formulation~\eqref{eq_MCSS} where agents take more than one measurement to reduce the amount of communications. In Case~(III) the SC network constraint in the  MCCN problem formulation~\eqref{eq_MCCN_bi} is relaxed and the minimum weight spanning forest is given as a solution, while the MCSS problem formulation~\eqref{eq_MCSS} is NP-hard in this case \cite{pequito_gsip} as some agents may share measurements from one or more parent SCCs.
We again mention that in this paper we consider minimum cost networked estimation accompanied with minimum number of measurements distinctly assigned to the agents.}

\section{Illustrative Example} \label{sec_example}
In this section, we provide an example to explain the methodology for Minimum-Cost Sensor Selection and Minimum-Cost Communication Network design. Consider an example system digraph with $18$ state nodes shown in Fig. \ref{fig_graph}.
\begin{figure}[!t]
	\centering
	\includegraphics[width=2.4in]{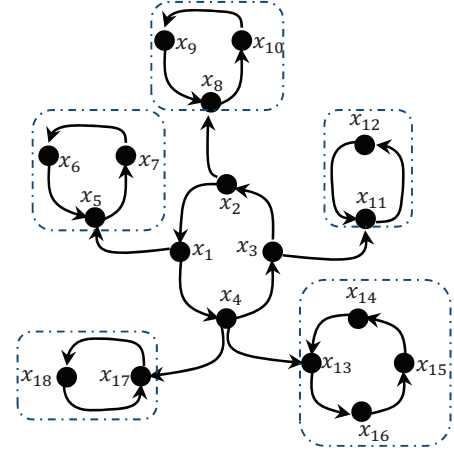}
	\caption{This figure shows a system digraph, where each node represents a state of the dynamical system. For simplicity the self-cycle at each node is not represented in the figure. The system contains $5$ parent SCCs (shown by dashed squares).}
	\label{fig_graph}
\end{figure}
This graph represents the structure of a system in the form \eqref{eq_sys1}. Note that we assume every state node in the graph contains a self-cycle, which is not shown for simplicity of the figure. Having  a self-cycle on every node the system is self-damped. Using the depth-first-search algorithm, it can be verified that the graph contains $6$ SCCs among which $5$ have no outgoing links and therefore are parent SCCs, marked by blue squares. Based on the Theorem~\ref{thm_Sp}, the minimum number of agents to estimate this system is $5$. Measurement of every state by each agent/sensor has a cost representing the matrix $\delta$. For this example, this agent-state measurement cost $\delta$ is randomly generated in the range $[0,10]$. To assign the states to be measured by agents, using~\eqref{eq_Delta}, the minimum cost state measurement in each parent SCC is considered to obtain cost matrix $\Delta$.  This agent-SCC measurement cost matrix is as follows:
\begin{eqnarray} \nonumber
\Delta = \left(
\begin{array}{ccccc}
8.1472  &  0.9754  &  1.5761  &  1.4189  &  6.5574 \\
9.0579 &   2.7850  &  9.7059 &   4.2176   & 0.3571 \\
1.2699  &  5.4688 &   9.5717  &  9.1574 &   8.4913 \\
9.1338  &  9.5751  &  4.8538  &  7.9221 &   9.3399 \\
6.3236  &  9.6489   & 8.0028  &  9.5949 &   6.7874
\end{array}
\right)
\end{eqnarray}
In order to solve the MCSS problem, based on the Formulation~\ref{prob_Z}, using Hungarian method the minimum-cost states in parent SCCs are assigned to the agents. This is done using MATLAB function \texttt{assignDetectionsToTracks}. The algorithm used by MATLAB is of complexity $\mc{O}(N^3)$, where~$N$ is the number of agents. The non-zero entries of optimal measurement structured matrix $\mc{C}$ are as follows: $\mc{C}(1,10)=1$, $\mc{C}(2,17)=1$, $\mc{C}(3,6)=1$, $\mc{C}(4,11)=1$, $\mc{C}(5,16)=1$.

For networked estimation/observability the communication network among these agents needs to be SC, as stated in Theorem~\ref{thm_net}.
In the communication network of agents the links are assumed to be bidirectional, and the symmetric communication cost matrix $\eta$ is considered randomly as follows:
\begin{eqnarray} \nonumber
\eta = \left(
\begin{array}{ccccc}
* &   7.2459 &   6.0784   & 5.4711 &   3.3048 \\
7.2459  &  *  &  4.8588  &  2.1386 &   2.7136 \\
6.0784  &  4.8588 &   *  &  8.5787  &  3.4038 \\
5.4711  &  2.1386  &  8.5787 &   *   & 4.4812 \\
3.3048   & 2.7136  &  3.4038 &   4.4812  &  *
\end{array}
\right)
\end{eqnarray}
To solve the MCCN problem (Problem Formulation~\ref{prob_mccn}) we use MATLAB function \texttt{graphminspantree}. The algorithm is of complexity $\mc{O}(N^2\log(N))$. The algorithm returns the non-zero entries of optimal communication network matrix $\mc{U}$ for this problem as follows: $\mc{U}(5,1)=\mc{U}(1,5)=1$, $\mc{U}(5,2)=\mc{U}(2,5)=1$, $\mc{U}(5,3)=\mc{U}(3,5)=1$, $\mc{U}(4,2)=\mc{U}(2,4)=1$.
This communication network of agents is shown in Fig.~\ref{fig_network}.
\begin{figure}[!t]
	\centering
	\includegraphics[width=1.5in]{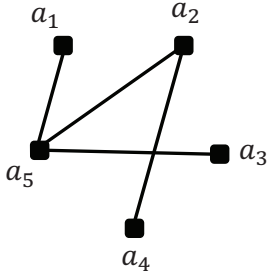}
	\caption{This figure shows the optimal communication network among $5$ agents measuring the parent SCCs in system digraph in Fig.~\ref{fig_graph}. The links are bidirectional and the graph represent the minimum cost spanning tree.}
	\label{fig_network}
\end{figure}
Note that sharing state-predictions over this communication network $\mc{U}$ among agents measuring states according to $\mc{C}$ results in an observable networked estimation of the self-damped system,  while the networked observability cost is minimized. All the algorithms used to optimally design the $0-1$ matrices $\mc{C}$ and $\mc{U}$ are structural and of polynomial time complexity.
\section{Conclusions} \label{sec_con}
It should be noted, following Remarks~\ref{mcss} and \ref{mccn}, the MCSS problem and the MCCN problem being generally NP-hard implies that the main MCNE problem in equation~\eqref{eq_formulate_1} is NP-hard as stated in Remark~\ref{mcne}.
However, based on Remarks~\ref{scc}, \ref{mcss2}, and \ref{mccn2} the MCNE problem for self-damped system constraint under bidirectional communication links is of computational complexity $\mc{O}(n^2+N^3)$ with $n$ as the number of state nodes (system size) and $N$ as the number of parent SCCs or agents (communication network size). If the number of agents is less than number of system states ($N<n^{\frac{2}{3}}$) the computational complexity of this problem is $\mc{O}(n^2)$.

%{
%Recall that in this paper we focus on minimal observability of the self-damped dynamical system. This is because having more number of measurements increases the sensing cost. In such case, having  more agents than what is required for minimal observability, the amount of network communications do not decrease. Although for this case the communication network does not need to be necessarily SC, but according to the proof of Lemma~\ref{lem_nec} we still need to connect  every agent measuring a parent SCC to every other agent without measurement of that parent SCC via a path.    }

{
Although in this paper we assume that the minimum number of measurements are each assigned to one agent, the solution can be extended to the case that every agent takes two (or more) distinct measurements. In such case, the communication network is smaller and the communication costs are less. We should emphasize that agents should take measurements from distinct parent SCCs, otherwise, in case agents share measurements of parent SCCs, the MCSS problem is NP-hard to solve \cite{pequito_gsip}. }
	
It should be noted that for general systems, i.e.,  systems that are not necessarily self-damped,  other than parent SCCs, contractions are the key components to ensure observability  \cite{doostmohammadian2017recovery,doostmohammadian2017observational}. Unlike parent SCCs, the contractions share nodes and therefore for such systems it is not possible to reformulate the MCSS problem as a linear assignment. One solution is to apply \textit{greedy algorithms}, which is the direction of our future research. Further, the communication network condition for networked observability also requires more than strong connectivity. For such systems, the network of agents  requires certain hubs measuring nodes in contractions along with SC network of agents measuring parent SCCs  \cite{globalsip14}. Therefore the MCCN problem is more complicated as it is our ongoing research.

\bibliographystyle{IEEEbib}
\bibliography{bibliography}

\begin{thebibliography}{10}

\bibitem{jstsp}
M.~Doostmohammadian and U.~Khan,
\newblock ``On the genericity properties in distributed estimation: Topology
  design and sensor placement,''
\newblock {\em IEEE Journal of Selected Topics in Signal Processing}, vol. 7,
  no. 2, pp. 195--204, 2013.

\bibitem{Sayed-LMS}
C.~G. Lopes and A.~H. Sayed,
\newblock ``Diffusion least-mean squares over adaptive networks: Formulation
  and performance analysis,''
\newblock {\em IEEE Transactions on Signal Processing}, vol. 56, no. 7, pp.
  3122--3136, July 2008.

\bibitem{acc13}
M.~Doostmohammadian and U.~A. Khan,
\newblock ``Topology design in networked estimation: a generic approach,''
\newblock in {\em American Control Conference}, Washington, DC, Jun. 2013, pp.
  4140--4145.

\bibitem{jstsp14}
M.~Doostmohammadian and U.~Khan,
\newblock ``Graph-theoretic distributed inference in social networks,''
\newblock {\em IEEE Journal of Selected Topics in Signal Processing}, vol. 8,
  no. 4, pp. 613--623, Aug. 2014.

\bibitem{nuno-suff.ness}
S.~Park and N.~Martins,
\newblock ``Necessary and sufficient conditions for the stabilizability of a
  class of {LTI} distributed observers,''
\newblock in {\em 51st IEEE Conference on Decision and Control}, 2012, pp.
  7431--7436.

\bibitem{park2017design}
S.~Park and N.~C. Martins,
\newblock ``Design of distributed {LTI} observers for state omniscience,''
\newblock {\em IEEE Transactions on Automatic Control}, vol. 62, no. 2, pp.
  561--576, 2017.

\bibitem{acc13_mesbahi}
A.~Chapman and M.~Mesbahi,
\newblock ``On strong structural controllability of networked systems: A
  constrained matching approach,''
\newblock in {\em American Control Conference}, Washington, DC, Jun. 2013, pp.
  6141--6146.

\bibitem{trefois2015zero}
M.~Trefois and J.~C. Delvenne,
\newblock ``Zero forcing number, constrained matchings and strong structural
  controllability,''
\newblock {\em Linear Algebra and its Applications}, vol. 484, pp. 199--218,
  2015.

\bibitem{chapman2015semi}
A.~Chapman,
\newblock {\em Semi-Autonomous Networks: Effective Control of Networked Systems
  through Protocols, Design, and Modeling},
\newblock Springer, 2015.

\bibitem{nowzari2016analysis}
C.~Nowzari, V.~M. Preciado, and G.~J. Pappas,
\newblock ``Analysis and control of epidemics: A survey of spreading processes
  on complex networks,''
\newblock {\em IEEE Control Systems}, vol. 36, no. 1, pp. 26--46, 2016.

\bibitem{may2001book}
R.~M. May,
\newblock {\em Stability and complexity in model ecosystems}, vol.~6,
\newblock Princeton University Press, 2001.

\bibitem{kailath1980linear}
T.~Kailath,
\newblock {\em Linear systems}, vol. 156,
\newblock Prentice-Hall Englewood Cliffs, NJ, 1980.

\bibitem{bar2004estimation}
Y.~Bar-Shalom, X.~R. Li, and T.~Kirubarajan,
\newblock {\em Estimation with applications to tracking and navigation: theory
  algorithms and software},
\newblock John Wiley \& Sons, 2004.

\bibitem{li2003survey}
X.~R. Li and V.~P. Jilkov,
\newblock ``Survey of maneuvering target tracking. part i. dynamic models,''
\newblock {\em IEEE Transactions on Aerospace and Electronic Systems}, vol. 39,
  no. 4, pp. 1333--1364, 2003.

\bibitem{wang2011target}
X.~Wang, M.~Fu, and H.~Zhang,
\newblock ``Target tracking in wireless sensor networks based on the
  combination of kf and mle using distance measurements,''
\newblock {\em IEEE Transactions on mobile computing}, vol. 11, no. 4, pp.
  567--576, 2011.

\bibitem{ennasr2018distributed}
O.~Ennasr and X.~Tan,
\newblock ``Distributed localization of a moving target: Structural
  observability-based convergence analysis,''
\newblock in {\em Annual American Control Conference (ACC)}. IEEE, 2018, pp.
  2897--2903.

\bibitem{el1992generalized}
F.~El-Hawary, F.~Aminzadeh, and G.~Mbamalu,
\newblock ``The generalized kalman filter approach to adaptive underwater
  target tracking,''
\newblock {\em IEEE Journal of Oceanic Engineering}, vol. 17, no. 1, pp.
  129--137, 1992.

\bibitem{pequito2017structurally}
S.~Kruzick, S.~Pequito, S.~Kar, J.~Moura, and A.~Aguiar,
\newblock ``Structurally observable distributed networks of agents under cost
  and robustness constraints,''
\newblock {\em IEEE Transactions on Signal and Information Processing over
  Networks}, 2017.

\bibitem{rowaihy2007survey}
H.~Rowaihy, S.~Eswaran, M.~Johnson, D.~Verma, A.~Bar-Noy, T.~Brown, and
  T.~La~Porta,
\newblock ``A survey of sensor selection schemes in wireless sensor networks,''
\newblock in {\em Unattended Ground, Sea, and Air Sensor Technologies and
  Applications IX}. International Society for Optics and Photonics, 2007, vol.
  6562, p. 65621A.

\bibitem{jan1993topological}
R.~Jan, F.~Hwang, and S.~Chen,
\newblock ``Topological optimization of a communication network subject to a
  reliability constraint,''
\newblock {\em IEEE Transactions on Reliability}, vol. 42, no. 1, pp. 63--70,
  1993.

\bibitem{wieselthier2000construction}
J.~E. Wieselthier, G.~D. Nguyen, and A.~Ephremides,
\newblock ``On the construction of energy-efficient broadcast and multicast
  trees in wireless networks,''
\newblock in {\em 19th Annual Joint Conference of the IEEE Computer and
  Communications Societies}, 2000, vol.~2, pp. 585--594.

\bibitem{baccelli1999poisson}
F.~Baccelli and S.~Zuyev,
\newblock ``Poisson-voronoi spanning trees with applications to the
  optimization of communication networks,''
\newblock {\em Operations Research}, vol. 47, no. 4, pp. 619--631, 1999.

\bibitem{pequito_gsip}
S.~Pequito, S.~Kar, and A.~P. Aguiar,
\newblock ``Minimum number of information gatherers to ensure full
  observability of a dynamic social network: a structural systems approach,''
\newblock in {\em IEEE 2nd Global Conference on Signal and Information
  Processing}, Atlanta, GA, Dec. 2014.

\bibitem{jiang2003optimal}
S.~Jiang, R.~Kumar, and H.~E. Garcia,
\newblock ``Optimal sensor selection for discrete-event systems with partial
  observation,''
\newblock {\em IEEE Transactions on Automatic Control}, vol. 48, no. 3, pp.
  369--381, 2003.

\bibitem{tzoumas2015minimal}
V.~Tzoumas, M.~A. Rahimian, G.~J. Pappas, and A.~Jadbabaie,
\newblock ``Minimal actuator placement with optimal control constraints,''
\newblock in {\em American Control Conference (ACC)}. IEEE, 2015, pp.
  2081--2086.

\bibitem{olshevsky2014minimal}
A.~Olshevsky,
\newblock ``Minimal controllability problems,''
\newblock {\em IEEE Transactions on Control of Network Systems}, vol. 1, no. 3,
  pp. 249--258, 2014.

\bibitem{krause2006near}
A.~Krause, C.~Guestrin, A.~Gupta, and J.~Kleinberg,
\newblock ``Near-optimal sensor placements: Maximizing information while
  minimizing communication cost,''
\newblock in {\em Proceedings of the 5th international conference on
  Information processing in sensor networks}. ACM, 2006, pp. 2--10.

\bibitem{pequito2014optimal}
S.~Pequito, F.~Rego, S.~Kar, A.~P. Aguiar, A.~Pascoal, and C.~Jones,
\newblock ``Optimal design of observable multi-agent networks: A structural
  system approach,''
\newblock in {\em IEEE European Control Conference}, 2014, pp. 1536--1541.

\bibitem{pequito2014design}
S.~Pequito, S.~Kar, S.~Sundaram, and A.~P. Aguiar,
\newblock ``Design of communication networks for distributed computation with
  privacy guarantees,''
\newblock in {\em IEEE 53rd Annual Conference on Decision and Control (CDC)},
  2014, pp. 1370--1376.

\bibitem{chen1989strongly}
W.~T. Chen and N.~F. Huang,
\newblock ``The strongly connecting problem on multihop packet radio
  networks,''
\newblock {\em IEEE Transactions on Communications}, vol. 37, no. 3, pp.
  293--295, 1989.

\bibitem{globalsip14}
M.~Doostmohammadian and U.~A. Khan,
\newblock ``On the characterization of distributed observability from first
  principles,''
\newblock in {\em 2nd IEEE Global Conference on Signal and Information
  Processing}, 2014, pp. 914--917.

\bibitem{rein_book}
K.~J. Reinschke,
\newblock {\em Multivariable control, a graph theoretic approach},
\newblock Berlin: Springer, 1988.

\bibitem{woude:03}
J.~M. Dion, C.~Commault, and J.~van~der Woude,
\newblock ``Generic properties and control of linear structured systems: {A}
  survey,''
\newblock {\em Automatica}, vol. 39, pp. 1125--1144, Mar. 2003.

\bibitem{liu-pnas}
Y.~Y. Liu, J.~J. Slotine, and A.~L. Barab\'{a}si,
\newblock ``Observability of complex systems,''
\newblock {\em Proceedings of the National Academy of Sciences}, vol. 110, no.
  7, pp. 2460--2465, 2013.

\bibitem{asilomar11}
M.~Doostmohammadian and U.~A. Khan,
\newblock ``Communication strategies to ensure generic networked observability
  in multi-agent systems,''
\newblock in {\em 45th Annual Asilomar Conference on Signals, Systems, and
  Computers}, Pacific Grove, CA, Nov. 2011, pp. 1865--1868.

\bibitem{cormen2009introduction}
T.~H. Cormen, C.~E. Leiserson, R.~L. Rivest, and C.~Stein,
\newblock {\em Introduction to algorithms},
\newblock MIT press, 2009.

\bibitem{lin}
C.~Lin,
\newblock ``Structural controllability,''
\newblock {\em IEEE Transactions on Automatic Control}, vol. 19, no. 3, pp.
  201--208, Jun. 1974.

\bibitem{cattrysse1992survey}
D.~G. Cattrysse and L.~N. Van~Wassenhove,
\newblock ``A survey of algorithms for the generalized assignment problem,''
\newblock {\em European journal of operational research}, vol. 60, no. 3, pp.
  260--272, 1992.

\bibitem{burkard1999linear}
R.~E. Burkard and E.~Cela,
\newblock ``Linear assignment problems and extensions,''
\newblock in {\em Handbook of combinatorial optimization}, pp. 75--149.
  Springer, 1999.

\bibitem{kuhn1955hungarian}
H.~W. Kuhn,
\newblock ``The hungarian method for the assignment problem,''
\newblock {\em Naval Research Logistics (NRL)}, vol. 2, no. 1-2, pp. 83--97,
  1955.

\bibitem{eswaran1976augmentation}
K.~P. Eswaran and R.~E. Tarjan,
\newblock ``Augmentation problems,''
\newblock {\em SIAM Journal on Computing}, vol. 5, no. 4, pp. 653--665, 1976.

\bibitem{bang2008minimum}
J.~Bang-Jensen and A.~Yeo,
\newblock ``The minimum spanning strong subdigraph problem is fixed parameter
  tractable,''
\newblock {\em Discrete Applied Mathematics}, vol. 156, no. 15, pp. 2924--2929,
  2008.

\bibitem{khuller1995approximating}
S.~Khuller, B.~Raghavachari, and N.~Young,
\newblock ``Approximating the minimum equivalent digraph,''
\newblock {\em SIAM Journal on Computing}, vol. 24, no. 4, pp. 859--872, 1995.

\bibitem{khuller1996strongly}
S.~Khuller, B.~Raghavachari, and N.~Young,
\newblock ``On strongly connected digraphs with bounded cycle length,''
\newblock {\em Discrete Applied Mathematics}, vol. 69, no. 3, pp. 281--289,
  1996.

\bibitem{vetta2001approximating}
A.~Vetta,
\newblock ``Approximating the minimum strongly connected subgraph via a
  matching lower bound,''
\newblock in {\em SODA}, 2001, pp. 417--426.

\bibitem{prim1957shortest}
R.~C. Prim,
\newblock ``Shortest connection networks and some generalizations,''
\newblock {\em Bell Labs Technical Journal}, vol. 36, no. 6, pp. 1389--1401,
  1957.

\bibitem{kruskal1956shortest}
J.~B. Kruskal,
\newblock ``On the shortest spanning subtree of a graph and the traveling
  salesman problem,''
\newblock {\em Proceedings of the American Mathematical society}, vol. 7, no.
  1, pp. 48--50, 1956.

\bibitem{gallager1983distributed}
R.~G. Gallager, P.~A. Humblet, and P.~M. Spira,
\newblock ``A distributed algorithm for minimum-weight spanning trees,''
\newblock {\em ACM Transactions on Programming Languages and Systems}, vol. 5,
  no. 1, pp. 66--77, 1983.

\bibitem{algorithm}
T.~H. Cormen, C.~E. Leiserson, R.~L. Rivest, and C.~Stein,
\newblock {\em {Introduction to Algorithms}},
\newblock MIT Press, 2009.

\bibitem{matsui1995minimum}
T.~Matsui,
\newblock ``The minimum spanning tree problem on a planar graph,''
\newblock {\em Discrete Applied Mathematics}, vol. 58, no. 1, pp. 91--94, 1995.

\bibitem{doostmohammadian2017recovery}
M.~Doostmohammadian, H.~R. Rabiee, H.~Zarrabi, and U.~A. Khan,
\newblock ``Distributed estimation recovery under sensor failure,''
\newblock {\em IEEE Signal Processing Letters}, vol. 24, no. 10, pp.
  1532--1536, 2017.

\bibitem{doostmohammadian2017observational}
M.~Doostmohammadian, H.~R. Rabiee, H.~Zarrabi, and U.~Khan,
\newblock ``Observational equivalence in system estimation: Contractions in
  complex networks,''
\newblock {\em IEEE Transactions on Network Science and Engineering}, vol. 5,
  no. 3, pp. 212--224, 2018.

\end{thebibliography}

\end{document}